\documentclass[a4paper]{article}
\pdfoutput=1

\usepackage{hyperref}
\usepackage{amssymb,amsmath,amsthm}
\usepackage[T1]{fontenc}
\usepackage[english]{babel}
\usepackage[a4paper]{geometry}
\usepackage{microtype}
\usepackage{authblk}

\usepackage[dvipsnames]{xcolor} 

\title{Guarded Cubical Type Theory: \\
Path Equality for Guarded Recursion}

\author[1]{Lars Birkedal}
\author[1]{Ale\v{s} Bizjak}
\author[1]{Ranald Clouston}
\author[1]{Hans~Bugge~Grathwohl}
\author[1]{Bas~Spitters}
\author[2]{Andrea Vezzosi}
\affil[1]{Department of Computer Science, Aarhus University, Denmark}
\affil[2]{Department of Computer Science and Engineering, Chalmers University of Technology, Sweden}

\date{}

\usepackage{xparse}
\usepackage{bbm}
\usepackage{stmaryrd}
\usepackage{mathpartir}
\usepackage{tikz-cd}
\usepackage{listings}

\theoremstyle{plain}
\newtheorem{theorem}{Theorem}
\newtheorem{lemma}[theorem]{Lemma}
\newtheorem{corollary}[theorem]{Corollary}
\theoremstyle{definition}

\newtheorem{example}[theorem]{Example}
\theoremstyle{remark}

\DeclareDocumentCommand{\later}{ o m }{
  \IfNoValueTF{#1}
  {\mathord{\triangleright}#2}
  {\mathord{\triangleright}#1 . #2}}

\DeclareDocumentCommand{\pure}{ o m }{
  \IfNoValueTF{#1}
  {\term{next}#2}
  {\term{next} #1 . \, #2}}

\DeclareDocumentCommand{\gstream}{ o m }
{\IfNoValueTF{#1}
  {\ensuremath{\term{Str}_{#2}}}
  {\ensuremath{\term{Str}^{#1}_{#2}}}}

\DeclareDocumentCommand{\eqjudg}{ m m m o }{
    \IfNoValueTF{#4}
    {\ensuremath{#1 \vdash #2 = #3}}
    {\ensuremath{#1 \vdash #2 = #3 : #4}}
}

\DeclareDocumentCommand{\dfix}{ o m m }{
    \IfNoValueTF{#1}
    {\dfixEmp #2 . #3}
    {\dfixEmp^{#1} #2 . #3}
}

\DeclareDocumentCommand{\fix}{ o m m }{
    \IfNoValueTF{#1}
    {\fixEmp #2 . #3}
    {\fixEmp^{#1} #2 . #3}
}

\newcommand{\DD}{\mathbb{D}}
\newcommand{\Face}{\ensuremath{\mathbbm{F}}}

\newcommand{\Id}{\term{Id}}
\newcommand{\Ineg}[1]{\ensuremath{1-#1}}
\newcommand{\I}{\ensuremath{\mathbbm{I}}}
\newcommand{\LCompTy}[2]{\Phi(#1;#2)}
\newcommand{\LElCompEmp}{\operatorname{Comp}}
\newcommand{\LElComp}[1]{\LElCompEmp(#1)}
\newcommand{\LElEmp}{\operatorname{El}}
\newcommand{\LEl}[1]{\LElEmp(#1)}
\newcommand{\LFace}{\mathbbm{F}}
\newcommand{\LGlueEmp}{\operatorname{Glue}}
\newcommand{\LGlue}[4]{\LGlueEmp \, \left[#1 \mapsto (#2,#3)\right] \, #4}
\newcommand{\LIdop}{\operatorname{Id}}
\newcommand{\LId}[3]{\LIdop_{#1} (#2, #3 )}
\newcommand{\LUfw}{\LU_{f}^\omega}
\newcommand{\LUf}{\LU_{f}}
\newcommand{\LUw}{\mathcal{U}^\omega}
\newcommand{\LU}{\mathcal{U}}
\newcommand{\Lfaceto}[2]{\ensuremath{#1^{#2}}}
\newcommand{\Lface}[1]{\left[#1\right]}
\newcommand{\Lhastype}[3]{\ensuremath{#1 \vdash #2 : #3}}
\newcommand{\Llater}{\operatorname{\triangleright}}
\newcommand{\Lwftype}[2]{\ensuremath{#1 \vdash #2}}
\newcommand{\Nat}{\term{N}}
\newcommand{\OneF}{\ensuremath{1_{\Face}}}
\newcommand{\Path}{\term{Path}}
\newcommand{\UG}{\mathfrak{U}}
\newcommand{\U}{\term{U}}
\newcommand{\Ycom}{\term{Y}}
\newcommand{\ZeroF}{\ensuremath{0_{\Face}}}
\newcommand{\abs}[2]{\langle #1 \rangle #2}
\newcommand{\app}{\circledast}
\newcommand{\bfc}{\mathbf{c}}
\newcommand{\bnfeq}{::=}

\newcommand{\code}[1]{\widehat{#1}}
\newcommand{\comp}[4]{\term{comp}^#1\,#2~#3~#4}
\newcommand{\ctt}{\ensuremath{\mathsf{CTT}}}
\newcommand{\cube}{\mathcal{C}}
\newcommand{\dM}[1]{\mathbf{DM}\left(#1\right)}
\newcommand{\defeq}{\triangleq}
\newcommand{\dfixEmp}{\term{dfix}}
\newcommand{\dsubst}[3]{\ensuremath{\vdash #1 : #2\rightarrowtriangle #3}}
\newcommand{\fixEmp}{\term{fix}}
\newcommand{\fold}{\term{fold}}
\newcommand{\gctt}{\ensuremath{\mathsf{GCTT}}}
\newcommand{\gdtt}{\ensuremath{\mathsf{GDTT}}}
\newcommand{\hastype}[3]{\ensuremath{#1 \vdash #2 : #3}}
\newcommand{\hd}{\term{hd}}
\newcommand{\hrt}[1]{\left[ #1 \right]}
\newcommand{\isetsep}{\;\ifnum\currentgrouptype=16 \middle\fi|\;}
\newcommand{\join}{\vee}
\newcommand{\latercode}{\code{\triangleright}}
\newcommand{\meet}{\wedge}
\newcommand{\natrec}{\term{natrec}}
\newcommand{\ott}{\ensuremath{\mathsf{OTT}}}
\newcommand{\pathlambda}[1]{\mathop{\ensuremath{\langle #1 \rangle}}}
\newcommand{\prev}{\operatorname{prev}}
\newcommand{\psh}[1]{\ensuremath{\widehat{#1}}}
\newcommand{\subst}[2]{[#1/#2]}
\newcommand{\term}[1]{\ensuremath{\operatorname{\mathsf{#1}}}}

\newcommand{\totcube}{{\widehat{\cube \times \omega}}}
\newcommand{\unfold}{\term{unfold}}
\newcommand{\vrt}[1]{\left[\begin{array}{l} #1 \end{array}\right]}
\newcommand{\wfcxt}[1]{\ensuremath{#1 \vdash}}
\newcommand{\wftype}[2]{\ensuremath{#1 \vdash #2}}
\newcommand{\zero}{\term{0}}

\renewcommand{\phi}{\varphi}
\renewcommand{\succ}{\term{s}}

\theoremstyle{plain}
\newtheorem{proposition}[theorem]{Proposition}

\newcommand{\appendixurl}{\url{http://cs.au.dk/~birke/papers/gdtt-cubical-technical-appendix.pdf}}
\newcommand{\gctturl}{\url{http://github.com/hansbugge/cubicaltt/tree/gcubical}}

\begin{document}

\maketitle

\begin{abstract}
This paper improves the treatment of equality in guarded dependent type theory ($\gdtt$), by combining it with cubical type theory ($\ctt$).
$\gdtt$ is an extensional type theory with guarded recursive types, which are useful for building models of program logics, and for programming and reasoning with coinductive types.
We wish to implement $\gdtt$ with decidable type checking, while still supporting non-trivial equality proofs that reason about the extensions of guarded recursive constructions.
$\ctt$ is a variation of Martin-L\"of type theory in which the identity type is replaced by abstract paths between terms.
$\ctt$ provides a computational interpretation of functional extensionality, is conjectured to have decidable type checking, and has an implemented type checker.
Our new type theory, called guarded cubical type theory, provides a computational interpretation of extensionality for guarded recursive types.
This further expands the foundations of $\ctt$ as a basis for formalisation in mathematics and computer science.
We present examples to demonstrate the expressivity of our type theory, all of which have been checked using a prototype type-checker implementation, and present semantics in a presheaf category.
\end{abstract}

\section{Introduction}
\label{sec:intro}

Guarded recursion is a technique for defining and reasoning about infinite objects. Its applications include the definition of
productive operations on data structures more commonly defined via coinduction, such as streams, and the construction
of models of program logics for modern programming languages with features such as higher-order store and
concurrency~\cite{Birkedal:Step}. This is done via the type-former $\later$, called `later', which distinguishes
data which is available immediately from data only available after some computation, such as the unfolding of a
fixed-point. For example, guarded recursive streams are defined by the equation
\[
  \gstream{A} \;=\; A \times \later{\gstream{A}}
\]
rather than the more standard $\gstream{A} = A \times\gstream{A}$, to specify that the head is available now but
the tail only later. The type for fixed-point combinators is then $(\later A\to A)\to A$, rather than the logically inconsistent
$(A\to A)\to A$, disallowing unproductive definitions such as taking the fixed-point of the identity function.

Guarded recursive types were developed in a simply-typed setting by Clouston et al.~\cite{Clouston:Programming}, following
earlier work~\cite{Nakano:Modality,Atkey:Productive,Abel:Formalized}, alongside a logic for reasoning about such programs.
For large examples such as models of program logics, we would like to be able to formalise such reasoning. A major
approach to formalisation is via \emph{dependent types}, used for example in the proof assistants Coq~\cite{Coq:manual} and Agda~\cite{Norell:thesis}.
Bizjak et al.~\cite{Bizjak-et-al:GDTT}, following earlier
work~\cite{Birkedal+:topos-of-trees,Mogelberg:tt-productive-coprogramming}, introduced guarded dependent type theory
($\gdtt$), integrating the $\later$ type-former into a dependently typed calculus, and supporting the definition of guarded
recursive types as fixed-points of functions on universes, and guarded recursive operations on these types.

We wish to formalise non-trivial theorems about equality between guarded recursive constructions, but such
arguments often cannot be accommodated within \emph{intensional} Martin-L\"{o}f type theory. For example, we may need
to be able to reason about the extensions of streams in order to prove the equality of different stream functions. Hence
$\gdtt$ includes an equality reflection rule, which is well known to make type checking undecidable.
This problem is close to well-known problems with functional
extensionality~\cite[Sec. 3.1.3]{Hofmann:Extensional}, and indeed this analogy can be developed. Just as
functional extensionality involves mapping terms of type $(x:A)\to\Id B\, (fx)\, (gx)$ to proofs of $\Id\,(A\to B)\, f\, g$, extensionality
for guarded recursion requires an extensionality principle for later types, namely the ability to map terms of type
$\later\Id A\,t\,u$ to proofs of
$\Id\,(\later A)\,(\pure{t})\,(\pure{u})$, where $\pure{}$ is the constructor for $\later$. These types are isomorphic in the intended
model, the presheaf category $\psh{\omega}$ known as the \emph{topos of trees}, and so in $\gdtt$ their equality was asserted as an axiom.
But in a calculus without equality reflection we cannot merely assert such axioms without losing canonicity.

\emph{Cubical type theory} ($\ctt$)~\cite{Cubical} is a new type theory with a computational interpretation of functional
extensionality but without equality reflection, and hence is a candidate for extension with guarded recursion, so that we may
formalise our arguments without incurring the disadvantages of fully extensional identity
types. $\ctt$ was developed primarily to provide a computational interpretation of the
univalence axiom of Homotopy Type Theory~\cite{hottbook}. The most important novelty of $\ctt$ is the replacement of inductively defined identity types by \emph{paths}, which can be seen as maps from an abstract interval $\I$, and are introduced and eliminated much like functions. $\ctt$ can be extended with identity types which model all rules of standard Martin-L\"of type
theory~\cite[Sec. 9.1]{Cubical}, but these are equivalent to path types, and in our paper it suffices to work with
path types only.
$\ctt$ has sound denotational semantics in (fibrations in) \emph{cubical
  sets}, a presheaf category that is used to model homotopy types. Many basic
syntactic properties of $\ctt$, such as the decidability of type checking, and canonicity
for base types, are yet to be proved, but a type checker has been implemented%
\footnote{\url{https://github.com/mortberg/cubicaltt}} that confers some confidence in
such properties.

In Sec.~\ref{sec:type-theory-examples} of this paper we propose \emph{guarded cubical type theory} ($\gctt$), a
combination of the two type theories%
\footnote{with the exception of the \emph{clock quantification} of $\gdtt$, which we leave to future work.}
which supports non-trivial proofs about guarded recursive types via path equality, while retaining the potential for good
syntactic properties such as decidable type-checking and canonicity. In particular, just as a term can be defined in $\ctt$ to
witness functional extensionality, a term can be defined in $\gctt$ to witness extensionality for later types.
Further, we use elements of the interval of $\ctt$ to annotate fixed-points, and hence control their unfoldings.
This ensures that fixed-points are path equal, but not judgementally equal, to their unfoldings, and hence prevents
infinite unfoldings, an obvious source of non-termination in any calculus with infinite constructions.
The resulting calculus is shown via examples to be useful for reasoning about guarded recursive operations;
we also view it as potentially significant from the point of view of $\ctt$, extending its expressivity as a basis for formalisation.

In Sec.~\ref{sec:semantics} we give sound semantics to
this type theory via the presheaf category over the product of the categories used to define semantics for $\gdtt$ and $\ctt$.
This requires considerable work to ensure that the constructions of the two type theories remain sound in the new category,
particularly the glueing and universe of $\ctt$.
The key technical challenge is to ensure that the $\later$ type-former supports the \emph{compositions} that all types
must carry in the semantics of $\ctt$.

We have implemented a prototype type-checker for this extended type theory%
\footnote{\gctturl},
which provides confidence in the type theory's syntactic properties. All examples in this paper, and many others,
have been formalised in this type checker.

For reasons of space many details and proofs are omitted from this paper, but are included in a technical appendix\footnote{\appendixurl}.

\section{Guarded Cubical Type Theory}
\label{sec:type-theory-examples}

This section introduces guarded cubical type theory ($\gctt$), and presents examples of how it can be used to prove
properties of guarded recursive constructions.

\subsection{Cubical Type Theory}

We first give a brief overview of \emph{cubical type theory}%
\footnote{\url{http://www.cse.chalmers.se/~coquand/selfcontained.pdf} is a self-contained presentation of $\ctt$.} ($\ctt$)~\cite{Cubical}.
We start with a standard dependent type theory with $\Pi$, $\Sigma$, natural numbers, and a Russell-style universe:
\[
  \begin{array}{lcl@{\hspace{.2\linewidth}}l}
    \Gamma, \Delta & \bnfeq & () ~ | ~ \Gamma, x : A & \text{Contexts} \\[1ex]
    t,u,A,B  & \bnfeq & x  ~|~ \lambda x : A . t ~|~ t\, u ~|~ (x : A) \to B &\text{$\Pi$-types} \\
    & | & (t,u) ~|~ t.1 ~|~ t.2 ~|~ (x:A) \times B &\text{$\Sigma$-types} \\
    & | & \zero ~|~ \succ t ~|~ \natrec t \, u ~|~ \Nat & \text{Natural numbers} \\
    & | & \U & \text{Universe}
  \end{array}
\]
We adhere to the usual conventions of considering terms and types up to $\alpha$-equality, and writing $A \to B$, respectively
$A \times B$, for non-dependent $\Pi$ and $\Sigma$-types. We use the symbol `$=$' for judgemental equality.

The central novelty of $\ctt$ is its treatment of equality.
Instead of the inductively defined identity types of intensional Martin-L\"of type theory~\cite{Martin-Lof-1973}, $\ctt$ has
\emph{paths}.
The paths between two terms $t,u$ of type $A$ form a sort of function space, intuitively that of continuous maps from some interval $\I$ to $A$, with endpoints $t$ and $u$.
Rather than defining the interval $\I$ concretely as the unit interval $[0,1] \subseteq
\mathbb{R}$, it is defined as the \emph{free De Morgan algebra}
on a discrete infinite set of names $\{i, j, k, \dots \}$. A De Morgan algebra is a bounded distributive lattice with an involution $\Ineg{\cdot}$ satisfying the De Morgan laws
\begin{align*}
  \Ineg{(i\meet j)} &= (\Ineg{i}) \join (\Ineg{j}), & \Ineg{(i\join j)} &= (\Ineg{i}) \meet (\Ineg{j}).
\end{align*}
The interval $[0,1]  \subseteq \mathbb{R}$, with $\term{min}$, $\term{max}$ and $\Ineg{\cdot}$, is an example of a De Morgan algebra.

The syntax for elements of $\I$ is:
\[
  r, s ~ \bnfeq ~ 0 ~|~ 1 ~|~ i ~|~ \Ineg{r} ~|~ r \meet s ~|~ r \join s.
\]
$0$ and $1$ represent the endpoints of the interval.
We extend the definition of contexts to allow introduction of a new name:
\[
  \Gamma, \Delta ~\bnfeq~ \cdots ~|~ \Gamma, i:\I.
\]
The judgement $\hastype{\Gamma}{r}{\I}$ means that $r$ draws its names from $\Gamma$.
Despite this notation, $\I$ is not a first-class type.
Path types and their elements are defined by
the rules in Fig.~\ref{fig:path-typing-rules}.
\emph{Path abstraction}, $\pathlambda{i} t$, and \emph{path application}, $t \, r$, are analogous to $\lambda$-abstraction and function application, and support the familiar $\beta$-equality $(\pathlambda{i} t)\, r = t\subst{r}{i}$ and $\eta$-equality $\pathlambda{i} t\, i = t$.
There are two additional judgemental equalities for paths, regarding their endpoints: given $p : \Path A ~ t ~ u$ we have $p\, 0 = t$ and $p\, 1 = u$.

\begin{figure}[t]
  \begin{mathpar}
    \inferrule{%
      \wftype{\Gamma}{A} \\
      \hastype{\Gamma}{t}{A} \\
      \hastype{\Gamma}{u}{A}
    }{%
      \wftype{\Gamma}{\Path A ~ t ~ u}}
  \end{mathpar}
  \begin{mathpar}
    \inferrule{%
      \wftype{\Gamma}{A} \\
      \hastype{\Gamma, i:\I}{t}{A}
    }{%
      \hastype{\Gamma}{\pathlambda{i} t}{\Path A ~t\subst{0}{i} ~ t\subst{1}{i}}}
    \and
    \inferrule{%
      \hastype{\Gamma}{t}{\Path A ~ u ~ s} \\
      \hastype{\Gamma}{r}{\I}
    }{%
      \hastype{\Gamma}{t\, r}{A}}
  \end{mathpar}
  \caption{Typing rules for path types.}
  \label{fig:path-typing-rules}
\end{figure}

Paths provide a notion of identity which is more extensional than
that of intensional Martin-L\"of identity types, as exemplified by the proof term for functional extensionality:\label{funext}
\[
  \term{funext} \, f\, g \defeq \lambda p . \pathlambda{i} \lambda x .\, p \, x \, i
  ~ : ~
  \left((x : A) \to \Path B ~ (f\, x) ~ (g\, x)\right) \to \Path~ (A\to B) ~ f ~ g.
\]

The rules above suffice to ensure that path equality is reflexive, symmetric, and a congruence, but we also need it to be
transitive and, where the underlying type is the universe, to support a notion of transport.
This is done via \emph{(Kan) composition operations}.

To define these we need the \emph{face lattice}, $\Face$, defined as the free distributive lattice on the symbols $(i=0)$ and $(i=1)$ for all names $i$, quotiented by the relation $(i=0) \meet (i=1) = \ZeroF$.
The syntax for elements of $\Face$ is:
\[
  \phi,\psi ~\bnfeq~
  \ZeroF ~|~ \OneF ~|~ (i=0) ~|~ (i=1) ~|~ \phi \meet \psi ~|~ \phi \join \psi.
\]
As with the interval, $\Face$ is not a first-class type, but the judgement $\hastype{\Gamma}{\phi}{\Face}$ asserts
that $\phi$ draws its names from $\Gamma$.
We also have the judgement $\eqjudg{\Gamma}{\phi}{\psi}[\Face]$ which asserts the equality of $\phi$ and $\psi$ in the
face lattice. Contexts can be restricted by elements of $\Face$:
\[
  \Gamma, \Delta ~\bnfeq~ \cdots ~|~ \Gamma,\phi.
\]
Such a restriction affects equality judgements so that, for example, $\eqjudg{\Gamma,\phi}{\psi_1}{\psi_2}[\Face]$ is
equivalent to $\eqjudg{\Gamma}{\phi \meet \psi_1}{\phi \meet \psi_2}[\Face]$

We write $\hastype{\Gamma}{t}{A[\phi \mapsto u]}$ as an abbreviation for the two judgements $\hastype{\Gamma}{t}{A}$
and $\eqjudg{\Gamma,\phi}{t}{u}[A]$, noting the restriction with $\phi$ in the equality judgement.
Now the composition operator is defined by the typing and equality rule
\begin{mathpar}
  \inferrule{%
    \hastype{\Gamma}{\phi}{\Face} \\
    \wftype{\Gamma, i:\I}{A} \\
    \hastype{\Gamma, \phi, i : \I}{u}{A} \\
    \hastype{\Gamma}{a_0}{A\subst{0}{i}[\phi\mapsto u\subst{0}{i}]}
  }{%
    \hastype{\Gamma}{\comp{i}{A}{[\phi\mapsto u]}{a_0}}{A\subst{1}{i}[\phi\mapsto u\subst{1}{i}]}
  }.
\end{mathpar}
A simple use of composition is to implement the transport operation for $\Path$ types\label{transport-term}
\[
  \term{transp}^i \, A ~ a
  ~\defeq~
  \comp{i}{A}{[\ZeroF \mapsto []]}{a}
  ~:~
  A\subst{1}{i},
\]
where $a$ has type $A\subst{0}{i}$.
The notation $[]$ stands for an empty \emph{system}.  In general a system is a list
of pairs of faces and terms, and it defines an element of a type by giving the individual
components at each face. We extend the syntax as follows:
\[
    t,u,A,B ~\bnfeq~ \cdots ~|~ [ \phi_1~t_1,\ldots,\phi_n~t_n ].
\]
Below we see two of the rules for systems; they ensure
that the components of a system agree where the faces overlap, and that all the cases possible in the current context are
covered:
\begin{mathpar}
  \inferrule{%
    \wftype{\Gamma}{A}\\
    \eqjudg{\Gamma}{\phi_1\join\ldots\join\phi_n}{\OneF}[\Face] \\
    \hastype{\Gamma, \phi_i}{t_i}{A} \\
    \eqjudg{\Gamma,\phi_i\meet\phi_j}{t_i}{t_j}[A]\\ i,j=1 \ldots n \\
  }{%
    \hastype{\Gamma}{[ \phi_1~t_1,\ldots,\phi_n~t_n ]}{A}
  }
  \and
  \inferrule{%
    \hastype{\Gamma}{[ \phi_1~t_1,\ldots,\phi_n~t_n ]}{A}\\
    \eqjudg{\Gamma}{\phi_i}{\OneF}[\Face]
  }{%
    \eqjudg{\Gamma}{[ \phi_1~t_1,\ldots,\phi_n~t_n ]}{t_i}[A]
  }
\end{mathpar}
We will shorten $[\phi_1\join\ldots\join\phi_n \mapsto [ \phi_1~t_1,\ldots,\phi_n~t_n ]]$ to $[ \phi_1\mapsto t_1,\ldots,\phi_n\mapsto t_n ]$.

A non-trivial example of the use of systems is the proof that $\Path$ is transitive; given $p\,:\,\Path A~a~b$ and
$q\,:\,\Path A~b~c$ we can define   
\[
\term{transitivity}\,p\,q \defeq \pathlambda{i} \comp{j}{A}{[(i=0) \mapsto a, (i=1) \mapsto q\,j]}{(p \, i)} \, : \, \Path A~a~c.
\]
This builds a path between the appropriate endpoints because we have the equalities $\comp{j}{A}{[\OneF \mapsto a]}{(p \, 0)} = a$ and $\comp{j}{A}{[\OneF \mapsto q\,j]}{(p \, 1)} = q\,1 = c$.

For reasons of space we have omitted the descriptions of some features of $\ctt$, such as glueing, and the further
judgemental equalities for terms of the form $\comp{i}{A}{[\phi \mapsto u]}{a_0}$ that depend on the
structure of $A$.

\subsection{Later Types}\label{sec:later}

In Fig.~\ref{fig:typing-rules-later} we present the `later' types of guarded dependent type theory
($\gdtt$)~\cite{Bizjak-et-al:GDTT}, with judgemental equalities in Figs.~\ref{fig:ty-eq-rules-later}
and~\ref{fig:tm-eq-rules-later}. Note that we do not add any new equation for the interaction of compositions with $\later$;
such an equation would be necessary if we were to add the eliminator $\prev$ for $\later$, but this extension (which involves
clock quantifiers) is left to further work.
We delay the presentation of the fixed-point operation until the next section.

The typing rules use the \emph{delayed substitutions} of $\gdtt$, as defined in Fig.~\ref{fig:del-substs}.
Delayed substitutions resemble Haskell-style do-notation, or a delayed form of let-binding.
If we have a term $t:\later{A}$, we cannot access its contents `now', but if we are defining a type or term that itself has
some part that is available `later', then this part \emph{should} be able to use the contents of $t$.
Therefore delayed substitutions allow terms of type $\later{A}$ to be unwrapped by $\later$ and $\pure$.
As observed by Bizjak et al.~\cite{Bizjak-et-al:GDTT} these constructions generalise the \emph{applicative
functor}~\cite{McBride:Applicative} structure of `later' types, by the definitions $\term{pure} \, t \defeq \pure{t}$, and $f \app t \defeq
\pure[\hrt{f' \gets f, t' \gets t}]{f' \, t'}$, as well as a generalisation of the $\app$ operation from simple functions to
$\Pi$-types. We here make the new observation that delayed substitutions can express the function $\latercode:\later{\U}
\to\U$, introduced by Birkedal and M{\o}gelberg~\cite{Mogelberg:2013} to express guarded recursive types as fixed-points
on universes, as $\lambda u.\later[[u'\gets u]]{u'}$; see for example the definition of streams in Sec.~\ref{sec:streams}.

\begin{figure}
  \begin{mathpar}
    \inferrule{%
      \wfcxt{\Gamma}}{%
      \dsubst{\cdot}{\Gamma}{\cdot}}
    \and
    \inferrule{%
      \dsubst{\xi}{\Gamma}{\Gamma'} \\
      \hastype{\Gamma}{t}{\later[\xi]{A}}}{%
      \dsubst{\xi\hrt{x \gets t}}{\Gamma}{\Gamma', x:A}}
  \end{mathpar}
  \caption{Formation rules for delayed substitutions.}
  \label{fig:del-substs}
\end{figure}
\begin{figure}
  \begin{mathpar}
    \inferrule{%
      \wftype{\Gamma,\Gamma'}{A} \\
      \dsubst{\xi}{\Gamma}{\Gamma'}}{%
      \wftype{\Gamma}{\later[\xi]{A}}}
    \and
    \inferrule{%
      \hastype{\Gamma, \Gamma'}{A}{\U} \\
      \dsubst{\xi}{\Gamma}{\Gamma'}}{%
      \hastype{\Gamma}{\later[\xi]{A}}{\U}}
    \and
    \inferrule{%
      \hastype{\Gamma,\Gamma'}{t}{A} \\
      \dsubst{\xi}{\Gamma}{\Gamma'}}{%
      \hastype{\Gamma}{\pure[\xi]{t}}{\later[\xi]{A}}}
  \end{mathpar}
  \caption{Typing rules for later types.}
  \label{fig:typing-rules-later}
\end{figure}

\begin{figure}
    \begin{mathpar}
    \inferrule{%
      \dsubst{\xi\hrt{x\gets t}}{\Gamma}{\Gamma',x:B} \\
      \wftype{\Gamma,\Gamma'}{A}}{%
      \eqjudg{\Gamma}{\later[\xi\hrt{x \gets t}]{A}}{\later[\xi]{A}}}
    \and
    \inferrule{%
      \dsubst{\xi\hrt{x\gets t,y\gets u}\xi'}{\Gamma}{\Gamma',x:B,y:C,\Gamma''} \\
      \wftype{\Gamma,\Gamma'}{C} \\
      \wftype{\Gamma,\Gamma',x:B,y:C,\Gamma''}{A}}{%
      \eqjudg{\Gamma}{\later[\xi\hrt{x\gets t,y\gets u}\xi']{A}}
      {\later[\xi\hrt{y\gets u,x\gets t}\xi']{A}}}
    \and
    \inferrule{%
      \dsubst{\xi}{\Gamma}{\Gamma'} \\
      \wftype{\Gamma,\Gamma', x:B}{A} \\
      \hastype{\Gamma,\Gamma'}{t}{B}}{%
      \eqjudg{\Gamma}{\later[\xi\hrt{x\gets \pure[\xi]{t}}]{A}}
      {\later[\xi]{A\subst{t}{x}}}}
  \end{mathpar}

  \caption{Type equality rules for later types (congruence and equivalence rules are omitted).}
  \label{fig:ty-eq-rules-later}
\end{figure}

\begin{figure}
    \begin{mathpar}
    \inferrule{%
      \dsubst{\xi\hrt{x\gets t}}{\Gamma}{\Gamma', x:B} \\
      \hastype{\Gamma,\Gamma'}{u}{A} }{%
      \eqjudg{\Gamma}{%
        \pure[\xi\hrt{x\gets t}]{u}}{%
        \pure[\xi]{u}}[%
      \later[\xi]{A}]}
    \and
    \inferrule{%
      \dsubst{\xi\hrt{x\gets t,y\gets u}\xi'}{\Gamma}{\Gamma',x:B,y:C,\Gamma''} \\
      \wftype{\Gamma,\Gamma'}{C} \\
      \hastype{\Gamma,\Gamma',x:B,y:C,\Gamma''}{v}{A} }{%
      \eqjudg{\Gamma}{%
        \pure[\xi\hrt{x\gets t,y\gets u}\xi']{v}}{%
        \pure[\xi\hrt{y\gets u,x\gets t}\xi']{v}}[%
      \later[\xi\hrt{x\gets t,y\gets u}\xi']{A}]}
    \and
    \inferrule{%
      \dsubst{\xi}{\Gamma}{\Gamma'} \\
      \hastype{\Gamma,\Gamma',x:B}{u}{A} \\
      \hastype{\Gamma,\Gamma'}{t}{B}}{%
      \eqjudg{\Gamma}%
      {\pure[\xi\hrt{x \gets \pure[\xi]{t}}]{u}}%
      {\pure[\xi]{u\subst{t}{x}}}%
      [\later[\xi]{A\subst{t}{x}}] }
    \and
    \inferrule{%
      \hastype{\Gamma}{t}{\later[\xi]{A}} }{%
      \eqjudg{\Gamma}%
      {\pure[\xi\hrt{x\gets t}]{x}}%
      {t}%
      [\later[\xi]{A}]}
  \end{mathpar}
  \caption{Term equality rules for later types. We omit congruence and equivalence rules, and the rules for terms of type
  $\U$, which reflect the type equality rules of Fig.~\ref{fig:ty-eq-rules-later}.}
  \label{fig:tm-eq-rules-later}
\end{figure}

\begin{example}

In $\gdtt$ it is essential that we can convert terms of type $\later[\xi]{\term{Id}_A \, t~u}$ into terms of type $\term{Id}_{\later[\xi]{A}} \, (\pure[\xi]{t})~(\pure[\xi]{u})$, as it is essential for \emph{L\"ob induction}, the technique of proof by guarded recursion where we assume $\later{p}$, deduce $p$, and hence may conclude $p$ with no assumptions.
This is achieved in $\gdtt$ by postulating as an axiom the following judgemental equality:
\begin{equation}\label{eq:Id_and_later}
  \term{Id}_{\later[\xi]{A}} \, (\pure[\xi]{t})~(\pure[\xi]{u})
  \;=\;
  \later[\xi]{\term{Id}_A \, t~u}
\end{equation}
A term from left-to-right of \eqref{eq:Id_and_later} can be defined using the $\term{J}$-eliminator for identity types, but
the more useful direction is right-to-left, as proofs of equality by L\"ob induction involve assuming that we later have a
path, then converting this into a path on later types.
In fact in $\gctt$ we can define a term with the desired type:
\begin{equation}\label{eq:later_ext}
  \lambda p.\abs{i}{\pure[\xi[p'\gets p]]{p'\, i}} \;:\; (\later[\xi]{\Path A\, t\, u})\to
    \Path\,(\later[\xi]{A})\,(\pure[\xi]{t})\,(\pure[\xi]{u}).
\end{equation}
Note the similarity of this term and type with that of $\term{funext}$, for functional extensionality, presented on
page~\pageref{funext}. Indeed we claim that \eqref{eq:later_ext} provides a computational interpretation of extensionality
for later types.
\end{example}

\subsection{Fixed Points}
\label{sec:fix}

In this section we complete the presentation of $\gctt$ by addressing fixed points.
In $\gdtt$ there are fixed-point constructions $\fix{x}{t}$ with the judgemental equality $\fix{x}{t} =
t\subst{\pure{\fix{x}{t}}}{x}$.
In $\gctt$ we want decidable type checking, including decidable judgemental equality, and so
we cannot admit such an unrestricted unfolding rule.
Our solution it that fixed points should not be judgementally equal to their unfoldings, but merely \emph{path equal}.
We achieve this by decorating the fixed-point combinator with an interval element which specifies the position on this path.
The $0$-endpoint of the path is the stuck fixed-point term, while the $1$-endpoint is the same term unfolded once.
However this threatens canonicity for base types: if we allow stuck fixed-points in our calculus, we could have stuck closed terms $\fix[i]{x}{t}$ inhabiting $\Nat$.
To avoid this, we introduce the \emph{delayed} fixed-point combinator $\dfixEmp$, which produces a term `later' instead of a term `now'.
Its typing rule, and notion of equality, is given in Fig.~\ref{fig:dfix}.
We will write $\fix[r]{x}{t}$ for $t\subst{\dfix[r]{x}{t}}{x}$, $\fix{x}{t}$ for $\fix[0]{x}{t}$, and $\dfix{x}{t}$ for $\dfix[0]{x}{t}$.

\begin{figure}
\begin{mathpar}
  \inferrule{%
    \hastype{\Gamma}{r}{\I} \\
    \hastype{\Gamma, x : \later{A}}{t}{A}
  }{%
    \hastype{\Gamma}{\dfix[r]{x}{t}}{\later{A}}}
  \and
  \inferrule{%
    \hastype{\Gamma, x : \later{A}}{t}{A} }{%
    \eqjudg{\Gamma}%
    {\dfix[1]{x}{t}}%
    {\pure{t\subst{\dfix[0]{x}{t}}{x}}}%
    [\later{A}]}.
\end{mathpar}
  \caption{Typing and equality rules for the delayed fixed-point}
  \label{fig:dfix}
\end{figure}

\begin{lemma}[Canonical unfold lemma]
  \label{prop:unfold-lemma}
  For any term $\hastype{\Gamma, x : \later{A}}{t}{A}$ there is a path between $\fix{x}{t}$ and $t\subst{\pure{\fix{x}{t}}}{x}$, given by the term $\pathlambda{i} \fix[i]{x}{t}$.
\end{lemma}

Transitivity of paths (via compositions) ensures that $\fix{x}{t}$ is path equal to any number of fixed-point unfoldings of itself.

A term $a$ of type $A$ is said to be a \emph{guarded fixed point} of a function $f:\later{A}\to A$ if there is a
path from $a$ to $f(\pure{a})$.

\begin{proposition}[Unique guarded fixed points]
  \label{prop:unique-fix}
  Any guarded fixed-point $a$ of a term $f : \later{A} \to A$ is path equal to $\fix{x}{f\, x}$.
\end{proposition}
\begin{proof}
  Given $p : \Path A ~ a ~ (f \, (\pure{a}))$,
  we proceed by L\"ob induction, i.e., by assuming $\term{ih} : \later{(\Path A ~ a ~ (\fix{x}{f\, x}))}$.
  We can define a path 
  \[
    s\defeq \pathlambda{i} f (\pure[\hrt{q \gets \term{ih}}]{q\, i}) 
    ~:~
    \Path A ~ (f (\pure{a})) ~ (f (\pure{\fix{x}{f\, x}})),
  \]
  which is well-typed because the type of the variable $q$ ensures that $q\, 0$ is judgementally equal to $a$, resp. $q\,1$
  and $\fix{x}{f\, x}$.
  Note that we here implicitly use the extensionality principle for later \eqref{eq:later_ext}.
  We compose $s$ with $p$, and then with the inverse of the canonical unfold lemma of Lem.~\ref{prop:unfold-lemma},
  to obtain our path from $a$ to $\fix{x}{f\, x}$.
  We can write out our full proof term, where $p^{-1}$ is the inverse path of $p$, as
  \[
    \fix{\term{ih}}{\pathlambda{i} \comp{j}{A}{[(i=0) \mapsto p^{-1}, (i=1) \mapsto f (\dfix[\Ineg{j}]{x}{f\, x})]}{(f (\pure[\hrt{q \gets \term{ih}}]{q\, i}) )}}.
    \qedhere
  \]
\end{proof}

\subsection{Programming and Proving with Guarded Recursive Types}
In this section we show some simple examples of programming with guarded recursion, and prove properties of our
programs using L\"ob induction.

\subparagraph*{Streams.}\label{sec:streams}
The type of guarded recursive streams in $\gctt$, as with $\gdtt$, are defined as fixed points on the universe:
\[
  \gstream{A} \;\defeq\; \fix{x}{A\times \later[[y\gets x]]{y}}
\]
Note the use of a delayed substitution to transform a term of type $\later{\U}$ to one of type $\U$, as discussed at the start
of Sec.~\ref{sec:later}. Desugaring to restate this in terms of $\dfixEmp$, we have
\[
  \gstream{A} \;=\;
  A\times\later[[y\gets\dfix[0]{x}{A\times\later[[y\gets x]]{y}}]]{y}
\]
The head function $\hd:\gstream{A}\to A$ is the first projection.
The tail function, however, cannot be the second projection, since this yields a term of type
\begin{equation}\label{eq:streamtail0}
\later[\hrt{y \gets \dfix[0]{x}{A\times\later[\hrt{y\gets x}]{y}}}]{y}
\end{equation}
rather than the desired $\later{\gstream{A}}$.
However we are not far off; $\later{\gstream{A}}$ is judgementally equal to
$\later[\hrt{y \gets \dfix[1]{x}{A\times\later[\hrt{y\gets x}]{y}}}]{y}$, which is
the same term as \eqref{eq:streamtail0}, apart from endpoint $1$ replacing $0$.
The canonical unfold lemma (Lem.~\ref{prop:unfold-lemma}) tells us that we can build a path in $\U$ from $\gstream{A}$ to
$A \times \later{\gstream{A}}$; call this path $\abs{i}{\gstream{A}^i}$. Then we can transport between these types:
\[
  \term{unfold} \, s \defeq \term{transp}^i \, \gstream{A}^i \, s
  \qquad\qquad\qquad
  \term{fold} \, s \defeq \term{transp}^i \, \gstream{A}^{\Ineg{i}} \, s
\]
Note that the compositions of these two operations are path equal to identity functions, but not judgementally equal.
We can now obtain the desired tail function $\term{tl} : \gstream{A} \to \later{\gstream{A}}$ by composing the second projection with $\term{unfold}$, so $\term{tl} \, s \defeq (\term{unfold} \, s).2$.
Similarly we can define the stream constructor $\term{cons}$ (written infix as $::$) by using $\term{fold}$:
\[
  \term{cons} \defeq \lambda a, s . \term{fold} \, (a, s)
  ~:~
  A \to \later{\gstream{A}} \to \gstream{A}.
\]

We now turn to higher order functions on streams. We define $\term{zipWith} : (A \to B \to C) \to \gstream{A} \to
\gstream{B} \to \gstream{C}$, the stream function which maps a binary function on two input streams to produce an output stream, as
\begin{align*}
  \term{zipWith} \, f \defeq
  \fix{z}{\lambda s_1, s_2 . f \, (\term{hd} \, s_1) \, (\term{hd} \, s_2) \,::\,
  \pure[\vrt{z' \gets z \\ t_1 \gets \term{tl} \, s_1 \\ t_2 \gets \term{tl} \, s_2}]{z' \, t_1 \, t_2}}.
\end{align*}
Of course $\term{zipWith}$ is definable even with simple types and $\later$, but in $\gctt$ we can go further and prove properties
about the function:
\begin{proposition}[$\term{zipWith}$ preserves commutativity]\label{prop:zipwith-preserves-comm}
  If $f : A \to A \to B$ is commutative, then $\term{zipWith}\, f : \gstream{A} \to \gstream{A} \to \gstream{B}$ is commutative.
\end{proposition}
\begin{proof}
  Let $\term{c} : (a_1 : A) \to (a_2 : A) \to \Path B ~ (f\, a_1 \, a_2) ~ (f\, a_2 \, a_1)$ witness commutativity of $f$.
  We proceed by L\"ob induction, i.e., by assuming
  \[
    \term{ih} : \later{\left((s_1:\gstream{A})\to(s_2:\gstream{A})\to \Path B 
      ~ (\term{zipWith}\, f \, s_1 \, s_2)
      ~ (\term{zipWith}\, f \, s_2 \, s_1)\right)}.
  \]
  Let $i:\I$ be a fresh name, and $s_1, s_2 : \gstream{A}$.
  Our aim is to construct a stream $v$ which is $\term{zipWith}\, f \, s_1 \, s_2$ when substituting $0$ for $i$, and $\term{zipWith}\, f \, s_2 \, s_1$ when substituting $1$ for $i$.
  An initial attempt at this proof is the term
  \[
    v \,\defeq\,
    \term{c} \, (\term{hd}\, s_1) \, (\term{hd}\, s_2) \, i ~::~ 
    \pure[\vrt{q \gets \term{ih} \\
      t_1 \gets \term{tl}\, s_1 \\
      t_2 \gets \term{tl} \, s_2}]%
    {q \, t_1 \, t_2 \, i}
    ~:~
    \gstream{B},
  \]
  which is equal to 
  \[
    f \, (\term{hd} \, s_1) \, (\term{hd} \, s_2) ~::~
    \pure[\vrt{t_1 \gets \term{tl} \, s_1 \\ t_2 \gets \term{tl} \, s_2}]%
    {\term{zipWith} \, f \, t_1 \, t_2}
  \]
  when substituting $0$ for $i$, which is $\term{zipWith}\, f\, s_1\, s_2$, but \emph{unfolded once}.
  Similarly, $v\subst{1}{i}$ is $\term{zipWith}\, f\, s_2\, s_1$ unfolded once.
  Let $\abs{j}{\term{zipWith}^j}$ be the canonical unfold lemma associated with $\term{zipWith}$ (see Lem.~\ref{prop:unfold-lemma}).
  We can now finish the proof by composing $v$ with (the inverse of) the canonical unfold lemma.
  Diagrammatically, with $i$ along the horizontal axis and $j$ along the vertical:
  \begin{mathpar}
    \raisebox{-0.5\height}{\includegraphics{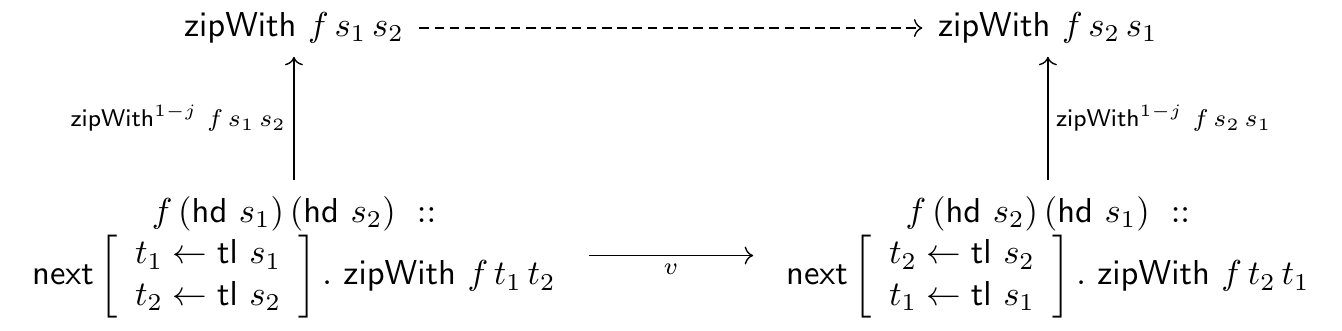}}
  \end{mathpar}
  The complete proof term, in the language of the type checker, can be found in Appendix~\ref{app:zipwith}.
\end{proof}

\subparagraph*{Guarded recursive types with negative variance.}
A key feature of guarded recursive types are that they support \emph{negative} occurrences of recursion variables.
This is important for applications to models of program logics~\cite{Birkedal:Step}.
Here we consider a simple example of a negative variance recursive type, namely
$  \term{Rec}_A \defeq \fix{x}{(\later[[x'\gets x]]{x'})\to A} $,
which is path equal to $\later{\term{Rec}_A} \to A$.
As a simple demonstration of the expressiveness we gain from negative guarded recursive types, we define a guarded variant of Curry's Y combinator:
\[
  \begin{array}{lclcl}
    \Delta &\defeq& \lambda x.f(\pure[[x'\gets x]]{((\unfold x') x})) &:& \later\term{Rec}_A\to A \\
    \Ycom &\defeq& \lambda f.\Delta(\pure\fold\Delta) &:& (\later A\to A)\to A,
  \end{array}
\]
where $\term{fold}$ and $\term{unfold}$ are the transports along the path between $\term{Rec}_A$ and $\later{\term{Rec}_A} \to A$.
As with $\term{zipWith}$, $\term{Y}$ can be defined with simple types and $\later$~\cite{Abel:Formalized}; what is new to $\gctt$ is that we can also prove properties about it:

\begin{proposition}[$\term{Y}$ is a guarded fixed-point combinator]
  \label{prop:Y-is-fixed-point-combinator}
  $\term{Y} f$ is path equal to $f \, (\pure({\term{Y} f}))$, for any $f : \later{A} \to A$. Therefore, by Prop.~\ref{prop:unique-fix}, $\term{Y}$ is path equal to $\term{fix}$.
\end{proposition}
\begin{proof}
  $\term{Y} f$ simplifies to $f\, (\pure{(\term{unfold} \, (\term{fold} \Delta) \, (\pure{\term{fold} \Delta}))})$, and $\term{unfold}\, (\term{fold} \Delta)$ is path equal to $\Delta$.
A congruence over this path yields our path between $\term{Y} f$ and
$f (\pure{(\term{Y} f)})$.
\end{proof}

\section{Semantics}
\label{sec:semantics}

In this section we sketch the semantics of $\gctt$.
The semantics is based on the category $\totcube$ of presheaves on the category $\cube \times \omega$, where $\cube$ is the \emph{category of cubes}~\cite{Cubical} and $\omega$ is the poset of natural numbers.
The category of cubes is the opposite of the Kleisli category of the free De Morgan algebra monad on finite sets.
More concretely, given a countably infinite set of names $i, j, k, \ldots$, $\cube$ has as objects finite sets of names $I$, $J$.
A morphism $I \to J \in \cube$ is a \emph{function} $J \to \dM{I}$, where $\dM{I}$ is the free De Morgan algebra with generators $I$.

Following the approach of Cohen et al.~\cite{Cubical}, contexts of $\gctt$ will be interpreted as objects of $\totcube$.
Types in context $\Gamma$ will be interpreted as pairs $(A, c_A)$ of a presheaf $A$ on the
category of elements of $\Gamma$ and a \emph{composition structure} $c_A$. We
call such a pair a \emph{fibrant} type.

To aid in defining what a composition structure is, and in showing that composition
structure is preserved by all the necessary type constructions,
we will make use of the internal language of $\totcube$ in the form of
\emph{dependent predicate logic}; see for example Phoa~\cite[App.~I]{phoa1992introduction}.

A type of $\gctt$ in context $\Gamma$ will then be interpreted as a pair of a type $\Lwftype{\Gamma}{A}$ in the internal language of $\totcube$, and a composition structure $c_A$, where $c_A$ is a term in the internal language of a specific type $\LCompTy{\Gamma}{A}$, which we define below after introducing the necessary constructs.
Terms of $\gctt$ will be interpreted as terms of the internal language.
We use \emph{categories with families}~\cite{dybjer1996internal} as our notion of a model.
Due to space limits we omit the precise definition of the category with families here, and refer to the online technical appendix.

The semantics is split into several parts, which provide semantics at different levels of generality.
\begin{enumerate}
\item We first show that every presheaf topos with a non-trivial internal De Morgan algebra $\I$ satisfying the disjunction property can be used to give semantics to the subset of the cubical type theory $\ctt$ without glueing and the universe.
  We further show that, for any category $\DD$, the category of presheaves on $\cube \times \DD$ has an interval $\I$, which is the inclusion of the interval in presheaves over the category of cubes $\cube$.
\item We then extend the semantics to include glueing and universes.
  We show that the topos of presheaves $\cube \times \DD$ for any category $\DD$ with an initial object can be used to give semantics to the entire cubical type theory.
\item Finally, we show that the category of presheaves on $\cube \times \omega$ gives semantics to delayed substitutions and fixed points.
  Using these and some additional properties of the delayed substitutions we show in the internal language of $\totcube$ that $\later[\xi]{A}$ has composition whenever $A$ has composition.
\end{enumerate}
Combining all three, we give semantics to $\gctt$ in $\totcube$.

\subsection{Model of $\ctt$ Without Glueing and the Universe}
\label{sec:model-of-the-basics}

Let $\mathcal{E}$ be a topos with a natural numbers object, and let $\I$ be a De Morgan algebra internal to $\mathcal{E}$
which satisfies the \emph{finitary disjunction property}, i.e.,
\[
  (i \join j) = 1 \implies (i = 1) \join (j = 1),
  \quad\text{and}\quad
  \lnot(0 = 1).
\]

\subparagraph*{Faces.}\label{sec:faces}

Using the interval $\I$ we define the type $\LFace$ as the image of the function $\cdot = 1 : \I \to \Omega$, where $\Omega$ is the subobject classifier.
More precisely, $\LFace$ is the subset type
\begin{align*}
  \LFace \defeq \left\{ p : \Omega \isetsep \exists (i : \I), p = (i = 1) \right\}
\end{align*}
We will implicitly use the inclusion $\LFace \to \Omega$.
The following lemma states in particular that the inclusion is compatible with all the lattice operations, so omitting it is justified.
The disjunction property is crucial for validity of this lemma.
\begin{lemma}\leavevmode
  \begin{itemize}
  \item $\Face$ is a lattice for operations inherited from $\Omega$.
  \item The corestriction $\cdot = 1 : \I \to \Face$ is a lattice homomorphism.
    It is not injective in general.
  \end{itemize}
\end{lemma}

Given $\hastype{\Gamma}{\phi}{\LFace}$, we write $\Lface{\phi} \defeq \LId{\LFace}{\phi}{\top}$.
Given $\Lwftype{\Gamma}{A}$ and $\Lhastype{\Gamma}{\phi}{\LFace}$ a \emph{partial element} of type $A$ of \emph{extent} $\phi$ is a term $t$ of type $\Lhastype{\Gamma}{t}{\Pi(p : \Lface{\phi}).A}$.
If we are in a context with $p : \Lface{\phi}$, then we will treat such a partial element $t$ as a term of type $A$, leaving implicit the application to the proof $p$, i.e., we will treat $t$ as $t\,p$.
We will often write $\Gamma, \Lface{\phi}$ instead of $\Gamma, p : \Lface{\phi}$ when we do not mention the proof term $p$ explicitly in the rest of the judgement.
This is justified since inhabitants of $\Lface{\phi}$ are unique up to judgemental equality (recall that dependent predicate logic is a logic over an extensional dependent type theory).
Given $\Lwftype{\Gamma, p : \Lface{\phi}}{B}$ we write $\Lfaceto{B}{\phi}$ for the dependent function space $\Pi(p : \Lface{\phi}).B$ and again leave the proof $p$ implicit.

For a term $\Gamma, p:\Lface{\phi} \vdash u : A$ we define $A[\phi \mapsto u] \defeq \Sigma (a:A).
\Lfaceto{\left(\LId{A}{a}{u}\right)}{\phi}$.

\subparagraph*{Compositions.}
Faces allow us to define the type of \emph{compositions} $\LCompTy{\Gamma}{A}$.
Homotopically, compositions allow us to put a lid on a box~\cite{Cubical}.
Given $\Lwftype{\Gamma}{A}$ we define the corresponding type of compositions as
\begin{align*}
  \LCompTy{\Gamma}{A} \defeq \Pi 
  &(\gamma : \I \to \Gamma)
  (\phi : \LFace)
  \left(u : \Pi (i:\I). \Lfaceto{\left(A(\gamma(i))\right)}{\phi}\right) . \\
  & A(\gamma(0))[ \phi \mapsto u(0) ] \to A(\gamma(1))[ \phi \mapsto u(1) ].
\end{align*}
Here we treat the context $\Gamma$ as a closed type.
This is justified because there is a canonical bijection between contexts and closed types of the internal language.
The notation $A(\gamma(i))$ means substitution along the (uncurried) $\gamma$.

Due to lack of space we do not show how the standard constructs of the type theory are interpreted.
We only sketch how the following composition term is interpreted in terms of the composition in the model.
\begin{mathpar}
  \inferrule{%
    \hastype{\Gamma}{\phi}{\Face} \\
    \wftype{\Gamma, i:\I}{A} \\
    \hastype{\Gamma, \phi, i : \I}{u}{A} \\
    \hastype{\Gamma}{a_0}{A\subst{0}{i}[\phi\mapsto u\subst{0}{i}]}
  }{%
    \hastype{\Gamma}{\comp{i}{A}{[\phi\mapsto u]}{a_0}}{A\subst{1}{i}[\phi\mapsto u\subst{1}{i}]}
  }.
\end{mathpar}
By assumption we have $c_A$ of type $\LCompTy{\Gamma, i : \I}{A}$ and $u$ and $a_0$ are interpreted as terms in the internal language of the corresponding types.
The interpretation of composition is the term
\begin{align*}
  \Lhastype{\gamma : \Gamma}{c_A \left(\lambda (i : \I) . (\gamma, i)\right)
                                 \phi
                                 \left(\lambda (i : \I) (p : \Lface{\phi}) . u\right)
                                 a_0}
                            {A(\gamma(1))[ \phi \mapsto u(1)]}
\end{align*}
where we have omitted writing the proof $u(0) = a_0$ on $\Lface{\phi}$.

\subparagraph*{Concrete models.}
The category of cubical sets has an internal interval type satisfying the disjunction property~\cite{Cubical}.
It is the functor mapping $I \in \cube$ to $\dM{I}$. Since the theory of a De Morgan
algebra with $0 \neq 1$ and the disjunction property is geometric~\cite[Section
X.$3$]{maclanemoerdijk92} we have that for any topos $\mathcal{F}$ and geometric morphism
$\phi:\mathcal{F}\to \psh{\cube}$, 
 $\phi^*(\I) \in \mathcal{F}$ is a De Morgan algebra with the disjunction
 property\footnote{A statement very close to this can be used as a characterisation of
   $\psh{\cube}$: this topos classifies the geometric theory of flat De Morgan algebras~\cite{Spitters:TYPES}.}.
In particular, given any category $\DD$ there is a projection functor $\pi : \cube \times
\DD \to \cube$ which induces the (essential) geometric morphism $\pi^* \dashv \pi_* : \psh{\cube \times \DD} \to \psh{\cube}$, where $\pi^*$ is precomposition with $\pi$, and $\pi_*$ takes limits along $\DD$.

\subparagraph*{Summary.}
With the semantic structures developed thus far we can give semantics to the subset of $\ctt$ without glueing and the
universe.

\subsection{Adding Glueing and the Universe}\label{sec:universes-and-glueing}
The glueing construction~\cite[Sec.~6]{Cubical} is used to prove both fibrancy and, subsequently, univalence of the universe of fibrant types.
Concretely, given
\[\Gamma \vdash \phi : \I\qquad
  \Gamma, [\phi] \vdash T\qquad
  \Gamma \vdash A\qquad
  \Gamma \vdash w : (T \to A)^{\phi}\]
we define the type $\LGlue{\phi}{T}{w}{A}$ in two steps.
First we define the type\footnote{This type is already present in Kapulkin at al.~\cite[Thm\ 3.4.1]{kapulkin2012simplicial}.}
\[
  Glue'_\Gamma(\phi, T, A, w) \defeq \sum_{a : A}\sum_{t : T^{\phi}}\prod_{p : [\phi]} w p (t p) = a.
\]
For this type we have the following property $\Gamma, [\phi] \vdash T \cong Glue'_\Gamma(\phi, T, A, w)$.
However, we need an equality, not an isomorphism, to obtain the correct typing rules.
The technical appendix provides a general strictification lemma which allows us to define
the type $Glue$.

To show that the type $\LGlue{\phi}{T}{w}{A}$ is fibrant we need to additionally assume that the map $\phi \mapsto \lambda\_.\phi : \LFace \to (\I \to \LFace)$ has an internal right adjoint $\forall$.
Such a right adjoint exists in all toposes $\psh{\cube\times\DD}$, for any small category $\DD$ with an initial object.

\subparagraph*{Universe of fibrant types.}
Given a (Grothendieck) universe $\UG$ in the meta-theory, the Hofmann-Streicher universe~\cite{Hofmann-Streicher:lifting} $\LUw$ in $\totcube$ maps $(I,n)$ to the set of functors valued in $\UG$ on the category of elements of $y(I,n)$,
where $y$ is the Yoneda embedding.
As in Cohen et al.~\cite{Cubical} we define the universe of \emph{fibrant} types $\LUfw$ by setting $\LUfw(I,n)$ to be the set of fibrant types in context $y(I,n)$.
The universe $\LUfw$ satisfies the rules
\begin{mathpar}
  \inferrule{%
    \Lhastype{\Gamma}{a}{\LU} \\
    \Lhastype{}{\bfc}{\LCompTy{\Gamma}{\LEl{a}}}}{%
    \Lhastype{\Gamma}{\llparenthesis a, \bfc \rrparenthesis}{\LUf}}
  \and
  \inferrule{%
    \Lhastype{\Gamma}{a}{\LUf}}{%
    \Lwftype{\Gamma}{\LEl{a}}}
  \and
  \inferrule{%
    \Lhastype{\Gamma}{a}{\LUf}}{%
    \Lhastype{}{\LElComp{a}}{\LCompTy{\Gamma}{\LEl{a}}}}
\end{mathpar}
Using the glueing operation, one shows that the universe of fibrant types is \emph{itself} fibrant and,
moreover, that it is univalent.

\subsection{Adding the Later Type-Former}
\label{sec:adding-the-later-modality}

We now fix the site to be $\cube \times \omega$.
From the previous sections we know that $\totcube$ gives semantics to $\ctt$.
The new constructs of $\gdtt$ are the $\Llater$ type-former and its delayed substitutions, and guarded fixed points.
Continuing to work in the internal language, we first show that the internal language of $\totcube$ can be extended with these
constructions, allowing interpretation of the subset of the type theory $\gdtt$ without clock quantification~\cite{Bizjak-et-al:GDTT}.
Due to lack of space we omit the details of this part, but do remark that $\Llater$ is defined as
\begin{align*}
  (\Llater(X))(I,n)
  \begin{cases}
    \{\star\} & \text{if } n = 0\\
    X(I,m) & \text{if } n = m+1
  \end{cases}
\end{align*}
The essence of this definition is that $\Llater$ depends only on the ``$\omega$ component'' and ignores the ``$\cube$ component''.
Verification that all the rules of $\gdtt$ are satisfied is therefore very similar to the verification that the topos $\psh{\omega}$ is a model of the same subset of $\gdtt$.

The only additional property we need now is that $\Llater$ preserves compositions, in the sense that if we have a delayed substitution $\dsubst{\xi}{\Gamma}{\Gamma'}$ and a type $\Lwftype{\Gamma,\Gamma'}{A}$ together with a closed term $\bfc_A$ of type $\LCompTy{\Gamma,\Gamma'}{A}$ then we can construct $\bfc'_{\later[\xi]{A}}$ of type $\LCompTy{\Gamma}{\later[\xi]{A}}$.

The following lemma uses the notion of a type $\Lwftype{\Gamma}{A}$ being \emph{constant with respect to $\omega$}.
This notion is a natural generalisation to types-in-context of the property that a presheaf is in the image of the functor $\pi^*$.
We refer to the online technical appendix for the precise definition.
Here we only remark that the interval type $\I$ is constant with respect to $\omega$, as is the type $\Lwftype{\Gamma}{\Lface{\phi}}$ for any term $\Lhastype{\Gamma}{\phi}{\LFace}$.

\begin{lemma}
  \label{lem:type-iso-pi-later-constant}
  Assume $\Lwftype{\Gamma}{A}$, $\Lwftype{\Gamma,\Gamma',x : A}{B}$ and $\dsubst{\xi}{\Gamma}{\Gamma'}$, and further that $A$ is constant with respect to $\omega$.
  Then the following two types are isomorphic
  \begin{align}
    \label{eq:type-iso-pi-later-constant}
    \Gamma \vdash \later[\xi]{\Pi(x : A).B} \cong \Pi(x : A).\later[\xi]B
  \end{align}
  and the canonical morphism $\lambda f . \lambda x . \pure[\hrt{\xi,f' \gets f}]{f'\,x}$ from left to right is an isomorphism.
\end{lemma}

\begin{corollary}
  \label{cor:later-pi-and-faces}
  If $\Lhastype{\Gamma}{\phi}{\LFace}$ then we have an isomorphism of types
  \begin{align}
    \Gamma \vdash \later[\xi]{\Pi(p : \Lface{\phi}).B} \cong \Pi(x : \Lface{\phi}).\later[\xi]{B}.
  \end{align}
\end{corollary}

\begin{lemma}[$\later{\xi}$-types preserve compositions]
  \label{lem:later-preserves-composition}
  If $\later[\xi]{A}$ is a well-formed type in context $\Gamma$ and we have a composition term $\bfc_A : \LCompTy{\Gamma,\Gamma'}{A}$, then there is a composition term $\bfc : \LCompTy{\Gamma}{\later[\xi]{A}}$.
\end{lemma}
\begin{proof}
  We show the special case with an empty delayed substitution.
  For the more general proof we refer to the technical appendix.
  Assume we have a composition $\bfc_A : \LCompTy{\Gamma}{A}$.
  Our goal is to find a term $\bfc : \LCompTy{\Gamma}{\later{A}}$, so we first introduce some variables:
  \begin{mathpar}
    \gamma : \I \to \Gamma \and
    \phi : \Face \and
    u : \Pi(i:\I).\left((\later{A}){(\gamma\, i)}\right)^{\phi} \and
    a_0 : (\later{A})(\gamma\, 0)[\phi \mapsto u\, 0].
  \end{mathpar}
  Using the isomorphisms from Cor.~\ref{cor:later-pi-and-faces} and Lem.~\ref{lem:type-iso-pi-later-constant} we obtain a term $\tilde{u} : \later{(\Pi(i : \I). (A(\gamma\, i))^{\phi})}$ isomorphic to $u$.
  We can now -- almost -- write the term
  \begin{equation}
    \tag{$*$}
    \label{eq:comp-for-later}
    \pure[\vrt{u' \gets \tilde{u} \\ a_0' \gets a_0}]{%
      \bfc_A \, \gamma \, \phi \, u' \, a_0'}
    ~:~ \later{(A(\gamma \, 1))},
  \end{equation}
  what is missing is to check that $a_0' = u'\, 0$ on the extent $\phi$, so that we can legally apply $\bfc_A$; this is equivalent to saying that the type
  $
  \later[\hrt{u' \gets \tilde{u}, a_0' \gets a_0}]{\LId{A(\gamma\, 0)}{a_0'}{u'\, 0}^{\phi}}
  $
  is inhabited.
  We transform this type as follows:
  \begin{align*}
    \later[\vrt{u' \gets \tilde{u} \\ a_0' \gets a_0}]{%
    \LId{}{a_0'}{u' \, 0}^{\phi}} 
 & \cong
   \left( \later[\vrt{u' \gets \tilde{u} \\ a_0' \gets a_0}]{%
    \LId{}{a_0'}{u' \, 0}}\right)^{\phi}
    \tag{Cor.~\ref{cor:later-pi-and-faces}}  \\
 & =
   \left( \LId{}{\pure[\vrt{u' \gets \tilde{u} \\ a_0' \gets a_0}]{a_0'}}{%
    \pure[\vrt{u' \gets \tilde{u} \\ a_0' \gets a_0}]{u'\, 0}} \right)^{\phi} \\
 & =
   \left(\LId{}{a_0}{u\, 0}\right)^{\phi},
  \end{align*}
  where the last equality uses that $\tilde{u}$ is defined using the inverse of $\lambda f \lambda x . \pure[\hrt{f' \gets f}]{f' \, x}$ (Lem.~\ref{lem:type-iso-pi-later-constant}).
  By assumption it is the case that $\left(\LId{}{a_0}{u\, 0}\right)^{\phi}$ is inhabited, and therefore (\ref{eq:comp-for-later}) is well-defined.
  It remains only to check that (\ref{eq:comp-for-later}) is equal to $u \, 1$ on the extent $\phi$, but this follows from the equalities of $\bfc_A$ and by the definition of $\tilde{u}$ (Lem.~\ref{lem:type-iso-pi-later-constant}). Assuming $\phi$, we have
  \[
    \pure[\vrt{u' \gets \tilde{u} \\ a_0' \gets a_0}]{%
      \bfc_A \, \gamma \, \phi \, u' \, a_0'}
    = 
    \pure[\vrt{u' \gets \tilde{u} \\ a_0' \gets a_0}]{%
      u' \, 1}
    =
    u \, 1.
    \qedhere
  \]
\end{proof}

\subparagraph*{Summary.}%
In this section we have highlighted the key ingredients that go into a sound interpretation of $\gctt$ in $\totcube$.
For the precise statement of the interpretation of all the constructs, and the soundness theorem, we refer to the online technical appendix.

\section{Conclusion}\label{sec:conclusion-future-work}
In this paper we have made the following contributions:
\begin{itemize}
  \item
We introduce guarded cubical type theory ($\gctt$), which combines features of cubical type theory
($\ctt$) and guarded dependent type theory ($\gdtt$).
The path equality of $\ctt$ is shown to support reasoning about extensional properties of guarded recursive operations,
and we use the interval of $\ctt$ to constrain the unfolding of fixed-points.
\item
We show that $\ctt$ can be modelled in any presheaf topos with an internal non-trivial De Morgan algebra with the disjunction property, an operator $\forall$, and a universe of fibrant types.
Most of these constructions are done via the internal logic.
We then show that a
class of presheaf models of the form $\psh{\cube \times \DD}$, for any category $\DD$ with an initial object,
satisfy the above axioms and hence gives rise to a model of $\ctt$.
\item
We give semantics to $\gctt$ in the topos of presheaves over $\cube \times \omega$.
\end{itemize}

\subparagraph{Further work.}%
We wish to establish key \emph{syntactic properties} of $\gctt$, namely decidable type-checking and canonicity for base
types. Our prototype implementation establishes some confidence in these properties.

We wish to further extend $\gctt$ with \emph{clock quantification}~\cite{Atkey:Productive}, such as is present in $\gdtt$.
Clock quantification allows for the controlled elimination of the later type-former, and hence the encoding of first-class coinductive types via
guarded recursive types.
The generality of our approach to semantics in this paper should allow us to build a model by combining cubical sets with the
presheaf model of $\gdtt$ with multiple clocks~\cite{Bizjak-Mogelberg:clock-sync}.
The main challenges lie in ensuring decidable type checking ($\gdtt$ relies on certain rules involving clock quantifiers which seem difficult to implement), and solving the \emph{coherence problem} for clock substitution.

Finally, some higher inductive types, like the truncation, can be added to $\ctt$.
We would like to understand how these interact with $\later$.

\subparagraph{Related work.}%
Another type theory with a computational interpretation of functional extensionality, but without equality reflection, is
observational type theory ($\ott$)~\cite{Altenkirch:Observational}.
We found $\ctt$'s prototype implementation, its presheaf semantics, and its interval as a tool for controlling unfoldings,
most convenient for developing our combination with $\gdtt$, but extending $\ott$ similarly
would provide an interesting comparison.

Spitters~\cite{Spitters:TYPES} used the interval of the internal logic of cubical sets to
model identity types. Coquand~\cite{Cubical:internal} defined the composition operation internally
to obtain a model of type theory. We have extended both these ideas to a full model of
$\ctt$.
Recent independent work by Orton and Pitts~\cite{Orton:Axioms} axiomatises a model
for $\ctt$ without a universe, again building on Coquand~\cite{Cubical:internal}.
With the exception of the absence of the universe, their development is more general than ours.
Our semantic developments are sufficiently general to support the sound addition of guarded recursive types to $\ctt$.

\subparagraph*{Acknowledgements.}

We gratefully acknowledge our discussions with Thierry Coquand, and the comments of our reviewers.
This research was supported in part by the ModuRes Sapere Aude Advanced Grant from The Danish Council for Independent Research for the Natural Sciences (FNU).
Ale\v{s} Bizjak was supported in part by a Microsoft Research PhD grant.

\bibliography{gCTT-trimmed}

\newpage
\appendix
\section{$\term{zipWith}$ Preserves Commutativity}
\label{app:zipwith}
We provide a formalisation of Sec.~\ref{sec:streams} which can be verified by our type checker.
This file, among other examples, is available in the \texttt{gctt-examples} folder in the type-checker repository.
\lstset{%
literate=     {\\ }{{$\lambda$}}1
              {<}{{$\langle$}}1
              {>}{{$\rangle$}}1
              {|>}{{$\triangleright$}}1
              {<-}{{$\leftarrow$}}1
              {->}{{$\rightarrow$}}1
}
\lstset{
basicstyle=\footnotesize\ttfamily,
morekeywords=[1]{module, data, where},
keywordstyle=[1]{\ttfamily\color{ForestGreen}},
morekeywords=[2]{Id,StrF,Str,StrUnfoldPath,unfoldStr,foldStr,cons,head,tail,zipWithF,zipWith,zipWithUnfoldPath,comm,zipWith_preserves_comm},
keywordstyle=[2]{\ttfamily\color{MidnightBlue}},
comment=[l]{--},
commentstyle={\ttfamily\color{RedOrange}},
}
\begin{lstlisting}
module zipWith_preserves_comm where

Id (A : U) (a0 a1 : A) : U = IdP (<i> A) a0 a1
data nat = Z | S (n : nat)

-- Streams of natural numbers
StrF (S : |> U) : U = (n : nat) * |> [S' <- S] S'

Str : U = fix (StrF Str)

-- The canonical unfold lemma for Str
StrUnfoldPath : Id U Str (StrF (next Str))
 = <i> StrF (dfix U StrF [(i=1)])

unfoldStr (s : Str) : (n : nat) * |> Str
 = transport StrUnfoldPath s

foldStr (s : (n : nat) * |> Str) : Str
 = transport (<i> StrUnfoldPath @ -i) s

cons (n : nat) (s : |> Str) : Str = foldStr (n, s)
head (s : Str) : nat = s.1
tail (s : Str) : |> Str = (unfoldStr s).2

-- Defining zipWith
zipWithF (f : nat -> nat -> nat) (rec : |> (Str -> Str -> Str)) 
 : Str -> Str -> Str
 = (\ (s1 s2 : Str) ->
      (cons (f (head s1) (head s2))
      	    (next [zipWith' <- rec, s1' <- tail s1 , s2' <- tail s2]
		  zipWith' s1' s2')))

zipWith (f : nat -> nat -> nat) : Str -> Str -> Str
 = fix (zipWithF f zipWith)

zipWithUnfoldPath (f : nat -> nat -> nat)
 : Id (Str -> Str -> Str)
      (zipWith f)
      (zipWithF f (next (zipWith f)))
 = <i> zipWithF f (dfix (Str -> Str -> Str) (zipWithF f) [(i=1)])

-- Commutativity property
comm (f : nat -> nat -> nat) : U = (m n : nat) -> Id nat (f m n) (f n m)

-- zipWith preserves commutativity.
zipWith_preserves_comm (f : nat -> nat -> nat) (c : comm f)
 : (s1 s2 : Str) -> Id Str (zipWith f s1 s2) (zipWith f s2 s1)
 = fix
   (\ (s1 s2 : Str) -> 
      <i> comp (<_> Str) 
               (cons (c (head s1) (head s2) @ i)
                     (next [q <- zipWith_preserves_comm
                           ,t1 <- tail s1
                           ,t2 <- tail s2]
                           q t1 t2 @ i))
               [(i=0) -> <j> zipWithUnfoldPath f @ -j s1 s2
               ,(i=1) -> <j> zipWithUnfoldPath f @ -j s2 s1])
\end{lstlisting}

\end{document}